\documentclass[letterpaper, 10 pt, conference]{ieeeconf}

\IEEEoverridecommandlockouts                                                                 
\title{\LARGE \bf A Game-Theoretic Approach to Robust Fusion \\
and  Kalman Filtering Under Unknown Correlations  }

\author{Spyridon Leonardos$^{1}$ and Kostas Daniilidis$^{2}$% <-this % stops a space
\thanks{$^{1,2}$The authors are with the  Department of Computer and Information Science, University of Pennsylvania
        Philadelphia, PA 19104, USA
        {\tt\small \{spyridon,kostas\}@seas.upenn.edu}}
}

\usepackage{amsmath,amsfonts,amssymb}
\usepackage{bm}
\usepackage{verbatim}
\usepackage{comment}
\usepackage{graphicx}
\usepackage{xcolor}
\usepackage{tikz}
\usetikzlibrary{arrows}
\usepackage{pgfplots}
\usepackage{cite}
\usepackage{hyperref}
\newcommand*{\QEDA}{\hfill\ensuremath{\blacksquare}}%
\newcommand{\inAppendix}{in  Appendix}
\newcommand{\inSection}{in Section }
\newcommand{\todo}[1]{\textcolor{red}{TODO: #1}}
\newcommand{\notion}[1]{\textit{#1}}

\newtheorem{thm}{Theorem}[section]

\newtheorem{definition}[thm]{Definition}

\newtheorem{lemma}[thm]{Lemma}

\newtheorem{remark}{Remark}

\newtheorem{problemstm}{\textbf{Problem Statement}}
\DeclareMathOperator*{\trace}{tr}

\newcommand{\mean}[1]{\overline{#1}}

\newcommand{\covs}[1]{\operatorname{Cov}\left( {#1} \right)}

\newcommand{\normal}{\mathcal{N}}
\newcommand{\xx}{{xx}}
\newcommand{\yy}{{yy}}
\newcommand{\xy}{{xy}}

\newcommand{\Sxy}{\Sigma_{xy}}
\newcommand{\Sxx}{\Sigma_{xx}}
\newcommand{\Syy}{\Sigma_{yy}}
\newcommand{\oSxx}{\widetilde{\Sigma}_{xx}}
\newcommand{\oSxy}{\widetilde{\Sigma}_{xy}}
\newcommand{\oSyy}{\widetilde{\Sigma}_{yy}}
\newcommand{\oSzz}{\widetilde{\Sigma}_{zz}}
\newcommand{\Szz}{\Sigma_{zz}}

\newcommand{\mapdef}[3]{#1 : #2 \rightarrow #3}

\newcommand{\reals}{\mathbb{R}}

\newcommand{\posdefcone}[1]{{S}^{#1}_+}

\newcommand{\opt}{^{\star}}

\newcommand{\norm}[1]{\lVert#1\rVert}

\usepackage{algorithm}
\usepackage{algpseudocode}
\algnewcommand\algorithmicinput{\textbf{INPUT:}}
\pgfplotsset{compat=newest}

\begin{document}

\maketitle
\thispagestyle{empty}
\pagestyle{empty}

%%%%%%%%%%%%%%%%%%%%%%%%%%%%%%%%%%%%%%%%%%%%%%%%%%%%%%%%%%%%%%%%%%%%%%%%%%%%%%%%
\begin{abstract}

This work addresses the problem of fusing two random vectors with unknown cross-correlations.
We present a formulation and a numerical method for computing the optimal estimate in the minimax sense.
We extend our formulation to linear measurement models that depend on two random vectors with unknown cross-correlations.
As an application we consider the problem of decentralized state estimation for a group of agents.
The proposed estimator takes cross-correlations into account while being less conservative than the widely used Covariance Intersection.
We demonstrate the superiority of the proposed method compared to Covariance Intersection with numerical examples and  simulations
within the specific application of decentralized state estimation using relative position measurements.

\end{abstract}

%%%%%%%%%%%%%%%%%%%%%%%%%%%%%%%%%%%%%%%%%%%%%%%%%%%%%%%%%%%%%%%%%%%%%%%%%%%%%%%%
\section{INTRODUCTION}

State estimation is one of the fundamentals problem in control theory and robotics. 
The most common state estimators are undoubtedly the Kalman filter \cite{kalman1960new}, which is optimal  in the minimum mean squared error for the case of linear systems, 
and its generalizations for nonlinear systems: the Extended Kalman Filter (EKF) \cite{sorenson1985kalman} and the Uscented Kalman Filter (UKF) \cite{julier1997new}.

In multi-agent systems, the task of state estimation takes a collaborative form in the sense that it involves inter-agent measurements and constraints.
Examples are cooperative localization in robotics \cite{Roumeliotis2002distributed} using relative pose measurements, 
camera network localization using epipolar constraints \cite{tron2014distributed} and many more.
On the one hand, a decentralized solution that scales with the number of agents is necessary.
On the other hand, the state estimates become highly correlated as information flows through the network.
Ignoring these correlations has grave consequences: estimates become optimistic and result in divergence of the estimator.
This phenomenon is analogous to the rumor spreading in social networks.

Unknown correlations may be present in other scenarios as well.
A popular simplification,  that significantly reduces computations and enables the use of  EKF-based estimators,  is that noise sources are independent. 
For instance, a common assumption in vision-aided inertial navigation is independence of the  image projection noises for each landmark \cite{mourikis2007multi,HernandezICRA15} although in reality, they are coupled with the motion of the camera sensor. 
In other cases, it might be impractical to store the entire covariance matrix due to storage limitations.
For instance, in Simultaneous Localization and Mapping (SLAM) problems, there have been several approaches that 
decouple the sensor state estimate  from the estimates of the landmark positions \cite{julier2007using,mourikis2007multi} to increase the efficiency of the estimator and reduce the storage requirements.

The most popular algorithm for fusion under the presence of unknown correlations is the Covariance Intersection (CI) method which was introduced by Julier and Uhlmann \cite{julier1997non}. In its simplest form, the Covariance Intersection algorithm is designed to fuse  two random vectors whose correlation is not known by forming a convex combination of the two estimates in the information space.
 Covariance Intersection produces estimates that are provably consistent, in the sense that estimated error covariance is an upper bound of the true error covariance.
However, it has been observed \cite{arambel2001covariance,xu2001application,nerurkar2007power} that Covariance Intersection produces estimates that are too conservative   
which may decrease the accuracy and convergence speed of the overall estimator when used as a component of an online estimator.

One of the most prominent applications of the proposed fusion algorithm is distributed state estimation in an EKF-based framework.
However, the problem of distributed state estimation is far from new.
There have been numerous approaches for EKF-based distributed state estimation and  EKF-based cooperative localization. 
Yet, some of them require that each agent maintains the state of the entire network \cite{Roumeliotis2002distributed,arambel2001covariance}, which is impractical and does not scale with the number of agents,  while others ignore correlations \cite{panzieri2006multirobot,martinelli2007improving} in order to simplify the estimation process or  use Covariance Intersection and variations of it \cite{li2013cooperative,carrillo2013decentralized} despite its slow convergence.

The contributions of this work are summarized as follows. First of all, we propose a method for fusion of two random vectors with unknown cross-correlations.
The proposed approach is less conservative than the widely used Covariance Intersection (CI) while taking cross-correlations into account.
Second of all, we extend our formulation  for the case of a linear measurement model.
Finally, we present numerical examples and simulations  in a distributed state estimation scenario which demonstrate the validity and comparative performance of the proposed approach  compared with the Covariance Intersection.

The paper is structured as follows: \inSection  \ref{sec:problemfrom} we include definitions of consistency and related notions and  we introduce the problem at hand.  The game-theoretic approach to fusing two random variables with unknown correlations is the topic of Section  \ref{sec:robustfusion} which is generalized for arbitrary linear measurement models \inSection \ref{sec:robustfusion}. In Section \ref{sec:numerical_opt} we include details on the implemented numerical algorithm. Numerical examples and simulation results are presented in Sections \ref{sec:numericexmaples} and \ref{sec:simulations} respectively.

% irrelevant stuff of this section:
%Minimax estimator: \cite{verdu1984minimax}

%%%%%%%%%%%%%%%%%%%%%%%%%%%%%%%%%%%%%%%%%%%%%%%%%%%%%%%%%%%%%%%%%%%%%%%%%%%%%%%%
\section{PROBLEM FORMALIZATION}

\label{sec:problemfrom}

In this section, we formalize the problem at hand.
First, we need a precise definition of \notion{consistency}.
 \begin{definition}[Consistency \cite{julier1997non}]
 Let $z$ be a  random vector with expectation $E[z]=\mean{z}$. An estimate $\widetilde{z}$ of  $\mean{z}$ is another random vector. The associated error covariance  is denoted $ \oSzz \doteq \covs{\widetilde{z}-\mean{z}}$. The pair $(\widetilde{z},\Szz)$  is  \textit{consistent} if  $E[\widetilde{z}] = \mean{z}$ and
\begin{equation}
\Szz \succeq  \oSzz 
\end{equation}
\end{definition}
 
\medskip

 \problemstm[Consistent fusion]{Given two consistent estimates $(\widetilde{x},\Sxx)$, $(\widetilde{y},\Syy)$ of $\mean{z}$, where  $\Sxx,\Syy$ are known upper bounds on the true error covariances.  
The problem at hand consists of fusing the two consistent estimates $(\widetilde{x},\Sxx)$, $(\widetilde{y},\Syy)$ in a single consistent estimate   $(\widetilde{z},\Szz)$, where $\widetilde{z}$ is of the form
\begin{equation}
\label{eq:weightedcomb}
\widetilde{z} = W_x \widetilde{x} + W_y \widetilde{y}
\end{equation}
with $W_x + W_y = I$ in order to preserve the mean.  
}

The most widely used solution of the above problem is the Covariance Intersection algorithm \cite{julier1997non}.
Given  upper bounds $\Sxx \succeq \oSxx$, $\Syy \succeq \oSyy$ the Covariance Intersection  equations read
\begin{equation}
\label{eq:cibasic}
\begin{split}
 \widetilde{z} &= \Szz \left \lbrace  \omega  \Sxx^{-1}\widetilde{x} + (1-\omega)  \Syy^{-1}\widetilde{y}  \right \rbrace \\
 \Szz^{-1} &= \omega  \Sxx^{-1} + (1-\omega)  \Syy^{-1}
\end{split}
 \end{equation}
where $\omega \in [0,1]$.
It can be immediately seen that $\Szz \left \lbrace  \omega  \Sxx^{-1}  + (1-\omega)  \Syy^{-1} \right \rbrace = I$ which implies $E[\widetilde{z}] = \mean{z}$.
Moreover, it is easy to check that $( \widetilde{z},\Szz)$ is consistent.  
The above can be easily generalized  for the case of more than 2 random variables, for partial measurements and for the linear measurement model we consider in Section \ref{sec:rkf}. Usually, $\omega$ is chosen such that either $\trace(\Szz)$  or $\log \det (\Szz^{-1})$ is minimized.
%A drawback is that only a very small fraction of all possible weighting matrices $W_x,W_y$ is explored.
%which is the cause of the conservative nature of Covariance Intersection.

Next, we introduce a notion related to consistency but with relaxed requirements. 
Let $\posdefcone{n}$ denote the positive semidefinite cone, that is the set of $n \times n$ positive semidefinite matrices.
First, recall that a function $\mapdef{f}{\posdefcone{n}}{\reals}$ is called \notion{$\posdefcone{n}$-nondecreasing} \cite{boyd2004convex} if 
 \begin{equation}
  X \succeq Y \Rightarrow f(X) \geq f(Y)
 \end{equation}
 for any $X,Y \in \posdefcone{n}$. An example of such a function is $f(X)= \trace(X)$. Now, we are ready to introduce the notion of \notion{consistency with respect to a $\posdefcone{n}$-nondecreasing function}.
 \medskip
 
  \begin{definition}[$f$-Consistency]
Let $\mapdef{f}{\posdefcone{n}}{\reals}$ be a nondecreasing function (with respect to $\posdefcone{n}$) satisfying $f(\mathbf{0})=0$.
Let $z$ be a  random vector with expectation $E[z]=\mean{z}$ and $\widetilde{z}$ be an estimate of  $\mean{z}$ with associated error covariance  $ \oSzz$. The pair $(\widetilde{z},\Szz)$  is  \textit{$f$-consistent} if  $E[\widetilde{z}] = \mean{z}$ and
\begin{equation}
f(\Szz) \geq  f(\oSzz) 
\end{equation}
\end{definition}

 \medskip
 \remark{Observe that consistency implies $f$-consistency. However, the converse in not necessarily true. }

 \medskip
 \problemstm[Trace-consistent fusion]{
Given two consistent estimates $(\widetilde{x},\Sxx)$, $(\widetilde{y},\Syy)$ of $\mean{z}$, where  $\Sxx,\Syy$ are known upper bounds on the true error variances.  
The problem at hand consists of fusing the two consistent estimates $(\widetilde{x},\Sxx)$, $(\widetilde{y},\Syy)$ in a single trace-consistent estimate   $(\widetilde{z},\Szz)$, where $\widetilde{z}$ is a linear combination of $x$ and $y$ and  
\begin{equation}
\trace(\Szz) \geq   \trace(\oSzz) 
\end{equation}
}
Next, we introduce a game-theoretic formulation for the problem of trace-consistent fusion. Relaxing the consistency constraint to the trace-consistency constraint enables us to estimate the weighting matrices $W_x,W_y$ according to some optimality criterion, which is none other than the minimax of the trace of the covariance matrix.

\remark{No assumptions on the distribution of the estimates $\tilde{x}$ and $\tilde{y}$ have been made so far.}

%%%%%%%%%%%%%%%%%%%%%%%%%%%%%%%%%%%%%%%%%%%%%%%%%%%%%%%%%%%%%%%%%%%%%%%%%%%%%%%%%%%
\section{ROBUST FUSION}

\label{sec:robustfusion}

The goal of this section is the derivation of our minimax approach.
First, we need some basic notions from game theory.
A zero-sum, two-player game on $\reals^{m} \times \reals^{n}$ is defined by a pay-off function $\mapdef{f}{\reals^m \times \reals^n}{\reals}$.
 Intuitively, the first player makes a move $u \in \reals^m$ then, the second player makes a move $v \in \reals^m$ and receives payment from the first player equal to $f(u,v)$.
 The goal of the first player is to minimize its payment and the goal of the second player is to maximize the received payment.
 The game is \notion{convex-concave} if the pay-off function $f(u,v)$ is convex in $u$ for fixed $v$ and concave in $v$ for fixed $u$.
 For a review minimax and convex-concave games in the context of convex optimization, we refer the reader to \cite{ghosh2003minimax}.

Let $z$ be a  random vector with expectation $E[z]=\mean{z}$.
Assume we have two estimates $(\widetilde{x},\Sxx)$, $(\widetilde{y},\Syy)$ of $\mean{z}$ where  $\Sxx,\Syy$ are approximations to the true error covariances $\oSxx$, $\oSyy$.
Based on the discussion of Section \ref{sec:problemfrom}, the fused estimate is of the form
\begin{equation}
\widetilde{z} = (I-K) \widetilde{x} + K \widetilde{y}
\end{equation}
and the  associated error covariance $\oSzz \doteq \covs{\widetilde{z}-\mean{z}}$ is given by
\begin{equation}
\oSzz = 
\begin{bmatrix}
I-K & K
\end{bmatrix}
\begin{bmatrix}
 \oSxx &  \oSxy \\
  \oSxy^T & \oSyy
\end{bmatrix}
\begin{bmatrix}
I-K^{ T}  \\ K^{T}
\end{bmatrix}
\end{equation}
%where we have the following Linear Matrix Inequality (LMI) constraint on  $\oSxy$
%\begin{equation}
%\begin{bmatrix}
% \oSxx &  \oSxy \\
%  \oSxy^T & \oSyy
%\end{bmatrix}
%\succeq 0
%\end{equation}
However, $\oSxx$, $\oSyy$ are not known. Therefore, we define
\begin{equation}
\Szz \doteq 
\begin{bmatrix}
I-K & K
\end{bmatrix}
\begin{bmatrix}
 \Sxx &  \Sxy \\
  \Sxy^T & \Syy
\end{bmatrix}
\begin{bmatrix}
I-K^{ T}  \\ K^{T}
\end{bmatrix}
\end{equation}
where we have the following Linear Matrix Inequality (LMI) constraint on  $\Sxy$
\begin{equation}
\label{eq:validcor}
\begin{bmatrix}
 \Sxx &  \Sxy \\
  \Sxy^T & \Syy
\end{bmatrix}
\succeq 0
\end{equation}

\remark{It can be  seen that $\trace(\Szz)$ is convex in $K$ for a fixed $\Sxy$ satisfying \eqref{eq:validcor}.
Therefore, the supremum  of $\trace(\Szz)$ over all $\Sxy$ satisfying \eqref{eq:validcor} is a convex function of $K$.
Moreover, for a fixed $K$, $\trace(\Szz)$ is linear, and thus concave as well, in $\Sxy$ with a convex domain defined by \eqref{eq:validcor}.
It follows that  $\trace(\Szz)$ is a convex-concave function in $(K,\Sxy)$.
}

As anticipated, we formulate the problem of finding the weighting matrix $K$ as a zero-sum, two-player convex-concave game: the first player chooses $K$ to minimize $\trace(\Szz)$ whereas the second player chooses $\Sxy$ to maximize $\trace(\Szz)$. More specifically,
let $(K^\star,\Sxy^\star)$ be the solution to the following minimax optimization problem
\begin{equation}
\label{eq:fusion_convex}
\begin{aligned}
& \underset{K}{\text{minimize}} \ \ 
&  & \sup_{\Sigma_\xy } \ \  \trace(\Szz)\\
& \text{subject to}
& &  \begin{bmatrix}
 \Sigma_\xx &  \Sigma_\xy \\
  \Sigma_\xy^T & \Sigma_\yy
\end{bmatrix}
\succeq 0
\end{aligned}
\end{equation}
Then, the fused estimated and the associated error covariance are given by
\begin{equation}
\label{eq:updatesimple}
\begin{split}
\widetilde{z} &= (I-K^\star) \widetilde{x} + K^\star \widetilde{y} \\
\Szz^\star &= 
\begin{bmatrix}
I-K^\star  & K^\star
\end{bmatrix}
\begin{bmatrix}
 \Sigma_\xx &  \Sigma_\xy^\star \\
  \Sigma_\xy^{\star T} & \Sigma_\yy
\end{bmatrix}
\begin{bmatrix}
I-K^{\star T}  \\ K^{\star T}
\end{bmatrix}
\end{split}
\end{equation}

Naturally, we have the following lemma.
\begin{lemma}
\label{lem:simplefusion}
If $(\widetilde{x},\Sxx)$ and $(\widetilde{y},\Syy)$ are consistent, then the pair $(\widetilde{z},\Szz^\star)$ given by \eqref{eq:updatesimple} is trace-consistent.
\end{lemma}
A proof of lemma \ref{lem:simplefusion} is presented \inAppendix \ref{app:proofsimplefusion}.

The problem of numerically solving problem \eqref{eq:fusion_convex} is the topic of subsequent sections.
The case under consideration in this section can be viewed as a special case of the next section.

%%%%%%%%%%%%%%%%%%%%%%%%%%%%%%%%%%%%%%%%%%%%%%%%%%%%%%%%%%%%%%%%%%%%%%%%%%%%%%%%%%%
\section{ROBUST LINEAR UPDATE}

\label{sec:rkf}

In this section, we explore a more general setting. 
We assume we have two random vectors $x$, $y$ with expectations $E[x]=\mean{x}$ and $E[y]= \mean{y}$.
We have some estimates $\widetilde{x}$ and  $\widetilde{y}$ of $\mean{x}$ and $\mean{y}$ respectively with associated error covariances $\oSxx$ and $\oSyy$.
As before, we assume that the true error covariances are only approximately known. Let $\Sxx$ and $\Syy$ denote these approximate values.
We assume we have a linear measurement model of the form
\begin{equation}
z = C \mean{x} + D \mean{y} + \eta 
\end{equation}
where $\eta$ is a zero-mean noise process with covariance $\Sigma_\eta$.
We assume that the measurement noise process $\eta$ is independent to the estimates $\widetilde{x}$ and  $\widetilde{y}$.
As in the classic Kalman filter derivation, we are looking for  an update step of the form
\begin{equation}
\widetilde{x}^+ =  \widetilde{x} + K ( z-\widetilde{z})  
\end{equation}
where $\widetilde{z} \doteq C \widetilde{x} + D \widetilde{y}$.  The error of the update is given by
\begin{equation}
\widetilde{x}^+ -\mean{x}
       =  (I-KC)(\widetilde{x}-\mean{x}) -  K D ( \widetilde{y}-\mean{y})   +K \eta  
\end{equation}
and the associated error covariance  is defined as $\oSxx^+ \doteq \covs{\widetilde{x}^+ - \mean{x}}$ and is given by
 \begin{equation}
 \begin{split}
  \oSxx^+ & =
  \begin{bmatrix}
I-KC & -KD
\end{bmatrix}
\begin{bmatrix}
 \oSxx &  \oSxy \\
  \oSxy^T & \oSyy
\end{bmatrix}
\begin{bmatrix}
I-C^TK^{ T}  \\ -D^TK^{T}
\end{bmatrix}
\\
& \ + K \Sigma_\eta K^T
\end{split}
 \end{equation}
However, the true error covariances $ \oSxx$ and $ \oSyy$ are not known. Therefore, we define
 \begin{equation}
 \begin{split}
  \Sxx^+ & \doteq
  \begin{bmatrix}
I-KC & -KD
\end{bmatrix}
\begin{bmatrix}
 \Sxx &  \Sxy \\
  \Sxy^T & \Syy
\end{bmatrix}
\begin{bmatrix}
I-C^TK^{ T}  \\ -D^TK^{T}
\end{bmatrix}
\\
& \ + K \Sigma_\eta K^T
 \end{split}
 \end{equation}
where $\Sxy$ should satisfy \eqref{eq:validcor} in order to be a valid cross-correlation. 
To alleviate notation, let $X=K^T$ and define 
\begin{equation}
\label{eq:fdef}
 f(X,\Sxy) \doteq \trace (\Sigma_\xx^+)
\end{equation}
By rewriting \eqref{eq:validcor} using Schur complement, the minimax formulation is written as follows
 \begin{equation}
 \label{eq:problemgeneral}
\begin{aligned}
& \underset{X}{\text{minimize}} \ \ \sup_Q
& &  f(X,Q) \\
& \text{subject to}
& &     \Syy^{-1/2} Q^T  \Sxx^{-1} Q\Syy^{-1/2} - I\preceq 0
\end{aligned}
\end{equation}

Let $(X^\star, Q^\star)$ be the optimal solution of problem \eqref{eq:problemgeneral} and let  $(K^\star,\Sxy^\star) = (X^{\star T}, Q^\star)$. Then,
\begin{equation}
 \label{eq:fusiongeneral}
\begin{split}
 \widetilde{x}^+ &= (I-K^\star C) \widetilde{x} - K^\star D \widetilde{y} \\
  \Sxx^{+\star} &=
  \begin{bmatrix}
I-K^\star C & -K^\star  D
\end{bmatrix}
\begin{bmatrix}
 \Sxx &  \Sxy^\star \\
  \Sxy^{\star T} & \Syy
\end{bmatrix}
\begin{bmatrix}
I-C^TK^{\star T}  \\ -D^TK^{\star T}
\end{bmatrix}
\\
& \ + K \Sigma_\eta K^T
\end{split}
 \end{equation}

\begin{comment}

\begin{equation}
\begin{split}
  g(X) &= \trace \left(X^T  \left( C \Sigma_\xx C^T + D\Sigma_\yy D^T +  \Sigma_\eta   \right) X\right)  \\
  & -2 \trace\left( \left(  C\Sigma_\xx  \right)^T X \right) + \trace(\Sigma_\xx)  \\
 \end{split}
\end{equation}
and
\begin{equation}
  h(X,\Sxy) = 2 \trace \left( \left(   	D^T X X^T C -D^T X    \right)   \Sigma_\xy    \right)
\end{equation}
\end{comment}

%We have the constraint on the cross correlation $\Sigma_\xy$
%\begin{equation}
%\begin{bmatrix}
% \Sigma_\xx &  \Sigma_\xy \\
%  \Sigma_\xy^T & \Sigma_\yy
%\end{bmatrix}
%\succeq 0
%\end{equation}
%or equivalently by Schur complement
%\begin{equation}
% \Sigma_\yy-    \Sigma_\xy^T  \Sigma_\xx^{-1}  \Sigma_\xy \succeq 0
%\end{equation}

Naturally, we have the following lemma.
\begin{lemma}
\label{lem:rffusion}
If $(\widetilde{x},\Sxx)$ and $(\widetilde{y},\Syy)$ are consistent, then the pair $(\widetilde{x}^+,\Sxx^{+ \star})$ given by \eqref{eq:fusiongeneral} is trace-consistent
\end{lemma}
The proof of lemma \ref{lem:rffusion} is exactly analogous to the proof of lemma \ref{lem:simplefusion} presented \inAppendix \ref{app:proofsimplefusion}.
%A proof of lemma \ref{lem:rffusion} is presented \inAppendix \ref{app:proofrffusion}.

\remark{When $C=I$, $D=-I$, $\Sigma_\eta =0$, we recover the case of the simple fusion of two random vectors.}

\section{INTERIOR POINT METHODS FOR CONVEX-CONCAVE GAMES}
 
 \label{sec:numerical_opt}
 
In this section, we describe the numerical method we use to solve Problem \eqref{eq:problemgeneral}. First, we will look at the simpler case of an unconstrained  convex-concave game with pay-off function $f(u,v)$. A point $(u^\star,v^\star)$ is a \notion{saddle} point for an unconstrained convex-concave game with pay-off function $f(u,v)$ if
 \begin{equation}
  f(u^\star,v) \leq f(u^\star,v^\star) \leq f(u,v^\star)
 \end{equation}
and the optimality conditions for differentiable convex-concave pay-off function are
 \begin{equation}
  \nabla_u f(u^\star,v^\star) =0 , \qquad   \nabla_v f(u^\star,v^\star) =0
 \end{equation}
We   use the infeasible start Newton method  \cite{boyd2004convex}, outlined in Algorithm 1, to find the optimal solution of the unconstrained problem:
\begin{equation}
\underset{u}{\text{minimize}} \ \underset{u}{\text{maximize}} \ f(u,v)
\end{equation}
Intuitively, at each step the directions $\Delta u_{nt},\Delta v_{nt}$ are  the solutions of the first order approximation
\begin{equation}
0 = r(u+\Delta u_{nt},v+\Delta v_{nt}) \approx r(u,v)  +  D r(u,v)[ \Delta u_{nt} ,\Delta v_{nt}]
\end{equation}
where $r(u,v) =  [ \nabla_u f(u,v)^T,\nabla_v f(u,v)^T]^T $. Then, a backtracking line search is performed on the norm of the residual  along the previously computed directions.

 \begin{algorithm}
 \label{alg:infeasiblestart}
\caption{Infeasible start Newton method.}
\label{alg:DOI}
\begin{algorithmic}[1]
\Statex \hspace{-\algorithmicindent} \textbf{given:} starting points $u,v \in \mathbf{dom} f$,
\Statex \hspace{\algorithmicindent} tolerance $\epsilon >0$, $\alpha \in (0,1/2)$, $\beta \in (0,1)$.
\Statex \hspace{-\algorithmicindent} \textbf{Repeat}
\Statex 1. $r(u,v) =  [ \nabla_u f(u,v)^T,\nabla_v f(u,v)^T]^T $ % \statex for nonumber
\Statex 2. Compute Newton steps by solving 
\[ D r(u,v)[ \Delta u_{nt} ,\Delta v_{nt}]= - r(u,v)\]
\Statex  3. Backtracking line search on $\norm{r}_2$. % $t=1$
\Statex \hspace{\algorithmicindent} $t=1$.
\Statex \hspace{\algorithmicindent} $u_t=  u+t\Delta u_{nt}$, $v_t=  v+t\Delta v_{nt}$.
\Statex \hspace{\algorithmicindent} \textbf{While} {$\norm{r(u_t,v_t)}_2 > (1-\alpha t) \norm{r(u,v)}_2$}
\Statex \hspace{\algorithmicindent} \hspace{\algorithmicindent} $t= \beta t$.
\Statex \hspace{\algorithmicindent} \hspace{\algorithmicindent} $u_t=  u+t\Delta u_{nt}$, $v_t=  v+t\Delta v_{nt}$.
\Statex \hspace{\algorithmicindent} \textbf{EndWhile} 
\Statex 4. Update: $u=u+t\Delta u_{nt}$, $v=v+t\Delta v_{nt}$.
\Statex \hspace{-\algorithmicindent} \textbf{until} {$\norm{r(u,v)}_2 \leq \epsilon$}
\end{algorithmic}
\end{algorithm}

However, the problem at hand is slightly more complicated since it involves a linear matrix inequality.  Therefore, we use  the barrier method \cite{boyd2004convex}. Intuitively, a sequence of unconstrained  minimization  problems is solved, using the last point iteration is the starting point for the next iteration. Define for $t>0$, the cost function $f_t(X,Q)$ by 
 \begin{equation}
  f_t(X,Q) = t f(X,Q) + \log \det ( - f_1(Q))
 \end{equation}
 where $f(X,Q)$ as defined in \eqref{eq:fdef} and
 \begin{equation}
  f_1(Q) =  \Syy^{-1/2} Q^T  \Sxx^{-1} Q\Syy^{-1/2} - I
 \end{equation}
 Intuitively, $\frac{1}{t}f_t$ approaches $f$ as $t \rightarrow \infty$.
Note that $f_t(X,Q)$ is still convex-concave for $t >0$.  The optimality conditions for a fixed $t>0$ are given by
\begin{equation}
 \nabla_X f_t(X\opt ,Q\opt) = 0, \quad  \nabla_Q f_t(X\opt ,Q\opt) = 0 
\end{equation}
where explicit expressions for $\nabla_X f_t$ and $\nabla_Q f_t$ are presented \inAppendix \ref{app:newton} along with the linear equations  for computing $\Delta X_{nt},\Delta Q_{nt}$ .

Finally, the structure of the problem allows us to easily identify a strictly feasible initial point $(X_0,Q_0)$ where  $Q_0 = 0$ and $X_0$ is given by
 \begin{equation}
  (C \Sxx C^T + D\Syy D^T + \Sigma_\eta) X_0  = C \Sxx
 \end{equation}

 For details on the convergence of the infeasible start Newton method and the barrier method for convex-concave games, we refer the reader to \cite{ghosh2003minimax,boyd2004convex}.
 
  \remark{Notation: $Df(x)[h]$ denotes the \notion{(Fr\'echet) derivative} or \notion{differential} of $f$ at $x$ along $h$. Similarly, $Df(x,y)[h_x,h_y]$ denotes the \notion{differential} of $f$ at $(x,y)$ along $(h_x,h_y)$.}

\section{NUMERICAL EXAMPLES}

\label{sec:numericexmaples}

In this section, we present two numerical examples which shed light on the differences between the Covariance Intersection (CI) and the proposed Robust Fusion (RF) approaches.
First, consider the example of fusing two random variables with means $\widetilde{x} = \widetilde{y} = \begin{bmatrix} 0 & 0 \end{bmatrix}^T$ 
and covariances
\begin{equation}
\Sxx = \begin{bmatrix}
5 & 0 \\ 0 & 5 
\end{bmatrix}
, \
\Syy = 
\begin{bmatrix}
3 & 0 \\ 0 & 7 
\end{bmatrix}
\end{equation}
Let $(\widetilde{z}_{CI},\Sigma_{CI})$ and $(\widetilde{z}_{RF},\Sigma_{RF})$ be the fused estimates and the corresponding error covariances obtained from Covariance Intersection and Robust Fusion. We have that  $\widetilde{z}_{CI} = \widetilde{z}_{RF} = [0 \ 0]^T$ and
\begin{equation}
\Sigma_{CI} =  \begin{bmatrix}
3.79 & 0 \\ 0 & 5.79 
\end{bmatrix},
\
\Sigma_{RF} =  \begin{bmatrix}
3  & 0 \\ 0 & 5 
\end{bmatrix},
\end{equation}
%Although $\trace(\Sigma_{CI} )$ is less than each of $\trace(\Sxx)$ and $\trace(\Syy)$, we see that the produced upper bound on the error covariance is too conservative

%%%%%%%%%%%%%%%%%%%%%%%%%%%%%%%%%%%%%%%%%%%%%%%%%%%%%%%%%%%%%%%%%%%%%%%%%%%%%%%%%%

\begin{figure}[ht!]
\label{fig:MLE_CI_RF}
\begin{center}
\includegraphics[trim={5.8cm 17.4cm 4cm 4.3cm},clip,scale=0.85]{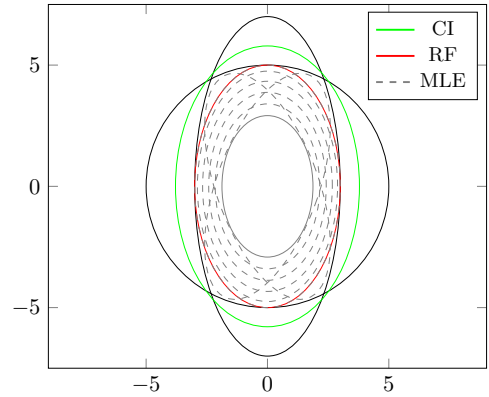}
%  trim={<left> <lower> <right> <upper>}

\caption{ Confidence ellipses: given a covariance matrix $\Sigma$ we draw the set $\{ x : x^T \Sigma^{-1}x =1\}$.
Initial confidence ellipse (black), Maximum likelihood Estimate (MLE) confidence ellipse (gray dashed) for various values of correlation, CI confidence ellipse (green) and RF confidence ellipse (red). 
The confidence ellipses  obtained from MLE lie in the intersection of the two ellipsoids $\{ x: x^T\Sxx^{-1}x \leq 1\}$ and $\{ x: x^T\Syy^{-1}x \leq 1\}$. 
Since the proposed approach is equivalent to MLE for some worst-case correlation, the RF confidence ellipse lies in the intersection of the two ellipsoids as well.
When correlation increases, the trace of the covariance of  MLE approaches the trace of $\Sigma_{RF}$.
CI is not maximum likelihood for any value of correlation but produces a guaranteed upper bound on the true error covariance.
}
\end{center}
\end{figure}

In the second example, we consider the case of partial measurements. More specifically, using notation of Section \ref{sec:rkf},  let 
\begin{equation}
\widetilde{x} = 
\begin{bmatrix}
0 \\ 0 
\end{bmatrix}
, \ 
\Sxx = \begin{bmatrix}
5 & 0 \\ 0 & 5 
\end{bmatrix}
\end{equation}
and $C = \begin{bmatrix}
1 & 0
\end{bmatrix}$, $z= \widetilde{z} = 0$, $ \Syy = 1$, $D=1$ and  $\Sigma_\eta = 0$.
Both Covariance Intersection and Robust Fusion yield $\widetilde{z}^+ = 0$ but
\begin{equation}
\Sigma_{CI} =  \begin{bmatrix}
3 & 0 \\ 0 & 6 
\end{bmatrix},
\
\Sigma_{RF} =  \begin{bmatrix}
1  & 0 \\ 0 & 5 
\end{bmatrix},
\end{equation}
Observe that despite we have a measurement  of only the first coordinate, the error variance of the second coordinate increased!
The reason for this phenomenon is that the CI updates the current estimate and the associated error covariance along a predefined direction only. 
Although $\trace(\Sigma_{CI}) < \trace(\Sxx)$, the bound on the true error covariance estimated by Covariance Intersection is very conservative.

\begin{figure}[ht!]
\label{fig:RF_CI_partial}
\begin{center}
\includegraphics[trim={5.8cm 17.4cm 4cm 4.3cm},clip,scale=0.85]{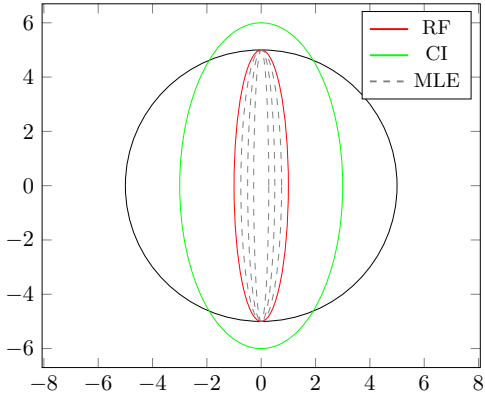}
%  trim={<left> <lower> <right> <upper>}
\caption{Illustration of the second numerical example. 
Initial confidence ellipse (black), Maximum likelihood Estimate (MLE) confidence ellipse (gray dashed) for various values of correlation, CI confidence ellipse (green) and RF confidence ellipse (red). 
}
\end{center}
\end{figure}

%%%%%%%%%%%%%%%%%%%%%%%%%%%%%%%%%%%%%%%%%%%%%%%%%%%%%%%%%%%%%%%%%%%%%%%%%%%%%%%%%%%%%%%%%%%%%%%%%%%%%%%%%%%%%%%%%%%%
%%%%%%%%%%%%%%%%%%%%%%%%%%%%%%%%%%%%%%%%%%%%%%%%%%%%%%%%%%%%%%%%%%%%%%%%%%%%%%%%%%%%%%%%%%%%%%%%%%%%%%%%%%%%%%%%%%%%
%%%%%%%%%%%%%%%%%%%%%%%%%%%%%%%%%%%%%%%%%%%%%%%%%%%%%%%%%%%%%%%%%%%%%%%%%%%%%%%%%%%%%%%%%%%%%%%%%%%%%%%%%%%%%%%%%%%%
\section{SIMULATIONS}

\label{sec:simulations}

Finally, we consider an application in distributed state estimation using relative position measurements. 
We experiment with a group of $n=4$ agents on the plane with a communication network topology as depicted in Fig. 3.
If there is an edge from  $i$ to $j$, then agent $i$ transmits its current state estimate and the corresponding error covariance estimate to agent $j$ which upon receipt, takes a measurement of the relative position and updates its own state estimate and associated error covariance estimate.
All agents have identical dynamics described by
\begin{equation}
 \begin{bmatrix}
  x_i(t+1) \\ v_i(t+1)
 \end{bmatrix}
=
\begin{bmatrix}
 I & I \\ 0 & I
\end{bmatrix}
 \begin{bmatrix}
  x_i(t) \\ v_i(t)
 \end{bmatrix}
 +
 \begin{bmatrix}
 0  \\  I
\end{bmatrix}
w_i(t)
\end{equation}
where $x_i(t) \in \reals^2$ and $v_i(t) \in \reals^2$ denote respectively the position and velocity of agent $i$ at time instance $t$ and
$w_i(t) \sim \normal (0,Q_i(t))$ is the noise process. If $\mathbf{x}_i(t) \doteq \begin{bmatrix} x_i(t)^T & v_i(t)^T \end{bmatrix}^T$, then let $A,B$ such that
\begin{equation}
 \mathbf{x}_i(t+1) = A \mathbf{x}_i(t) + B w_i(t)
\end{equation}

Only agent $1$ is equipped with  global position system (GPS), that is we have a measurement of the form
\begin{equation}
 y_1(t) = x_1(t) + \eta_1 (t)
\end{equation}
where $\eta_1(t) \sim \normal (0,R_1)$. Agent 1 performs a standard Kalman Filter update step after a GPS measurement.
For each edge $(i,j)$ we have a pairwise measurement of the form
\begin{equation}
 y_{ij}(t) = x_j(t)-x_i(t) + \eta_{ij} (t)
\end{equation}
where $\eta_{ij}(t) \sim \normal (0,R_{ij})$. Updates of the state estimates can be performed by either ignoring cross-correlations (Naive Fusion) or by one of Covariance Intersection or the proposed Robust Fusion.

Each agent maintains only its one state and communicates its to each neighbors at each time instance. 
The individual prediction step is the same as the Kalman Filter (KF) prediction step, that is
\begin{align}
  \widetilde{\mathbf{x}}_i(t+1|t) &= A \widetilde{\mathbf{x}}_i(t|t)  \\
  \Sigma_i(t+1|t) &=  A \Sigma_i(t+1|t)  A^T + B Q_i(t) B^T
\end{align}
where $\widetilde{\mathbf{x}}_i(t+1|t)$ denotes the estimate of agent $i$ for its state at time $t+1$ having received measurements up to time $t$ and $ \Sigma_i(t+1|t)$ is the associated error covariance.

We evaluate four estimator, three decentralized and one centralized: Naive Fusion (NF) which ignores correlations,  Robust Fusion (RF), Covariance Intersection (CI) and Centralized Kalman Filter (CKF). The Centralized Kalman Filter (CKF) is simply a standard Kalman Filter containing all agent states. It serves as a measure of how close the decentralized estimators are to the optimal centralized estimator.
Results can be seen in Figure 4 and Table I. We used the following values for the noise parameters:   $Q_i = 10^{-6} I_2$ for all agents, $R_1 =  I_2$ and $R_{ij} = 10^{-2} I_2$ for all pairwise measurements. 
The Robust Fusion  based estimator significantly outperforms the Covariance Intersection  based estimator which produces particularly noisy velocity estimates.

\begin{figure}
\label{fig:network_topology}
\begin{center}
\begin{tikzpicture}

\tikzset{vertex/.style = {shape=circle,draw,minimum size=1.5em}}
\tikzset{edge/.style = {->,> = latex'}}
% vertices
% 2.0000    0.6180   -1.6180   -1.6180    0.6180
% 0         1.9021    1.1756   -1.1756   -1.9021
\node[vertex] (a) at  (3,0) {$1$};
\node[vertex] (b) at  (3,3) {$2$};
\node[vertex] (c) at  (0,3) {$3$};
\node[vertex] (d) at  (0,0) {$4$};
%edges
\draw[edge] [bend right=15]  (a) to (b);
\draw[edge] (b) to (c);
\draw[edge] (b) to (d);
\draw[edge]  (c) to (d);
\draw[edge] [bend right=15] (d) to (a);
\draw[edge]   (c) to (a);
\draw[edge]  [bend right=20]  (b) to (a);
\draw[edge] [bend right=20] (a) to (d);
\end{tikzpicture}
\caption{Network topology.}
\end{center}
\end{figure}
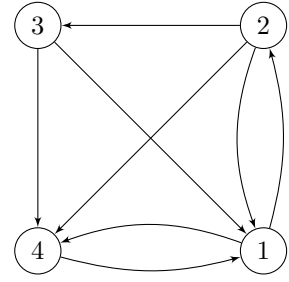

\def \figscale{0.44}
\def \labelx {\Large $t$}
\def \labely {\Large  Position error}
\def \labelyv {\Large Velocity error}
\begin{figure*}[ht!]
%\vspace{.2cm}
\label{fig:synthetic}
\begin{center}
\begin{tabular}{ccccc}

\rotatebox{90}{\hspace{1.0cm} \vspace{12cm}  Agent 1}
&
%%%%%%%%%%%%%%%%%%%%%%%%%%%%%%%%%%%%%%%%%%%%%%%%%%%%%%%%%%%%%%%%%%%%%%%%%%%%%%%%%%
\scalebox{\figscale}{
\begin{tikzpicture}
	\begin{axis}[
		xlabel={\labelx},
		ylabel={\labely},
		ymax = 7,
		xmax = 300,
		xmin =1,
		ymin=0
	]
	% use TeX as calculator:
	\addplot+[red,solid,mark=none] table[x=f,y=prf1] {simulationPos.txt}; \addlegendentry{\small RF}
	\addplot+[black,solid,mark=none] table[x=f,y=pnf1] {simulationPos.txt}; \addlegendentry{\small NF}
	\end{axis}
\end{tikzpicture}
}
&
\scalebox{\figscale}{
\begin{tikzpicture}
	\begin{axis}[
		xlabel={\labelx},
		ylabel={\labely},
		ymax = 2,
		xmax = 300,
		xmin =1,
		ymin=0
	]
	% use TeX as calculator:
	\addplot+[red,solid,mark=none] table[x=f,y=prf1] {simulationPos.txt}; \addlegendentry{\small RF}
	\addplot+[blue,solid,mark=none] table[x=f,y=pcf1] {simulationPos.txt}; \addlegendentry{\small CKF}
	\end{axis}
\end{tikzpicture}
}
&
\scalebox{\figscale}{
\begin{tikzpicture}
	\begin{axis}[
		xlabel={\labelx},
		ylabel={\labely},
		ymax = 2,
		xmax = 300,
		xmin =1,
		ymin=0
	]
	% use TeX as calculator:
	\addplot+[red,solid,mark=none] table[x=f,y=prf1] {simulationPos.txt}; \addlegendentry{\small RF}
	\addplot+[green,solid,mark=none] table[x=f,y=pci1] {simulationPos.txt}; \addlegendentry{\small CI}
	\end{axis}
\end{tikzpicture}
}
&
\scalebox{\figscale}{
\begin{tikzpicture}
	\begin{axis}[
		xlabel={\labelx},
		ylabel={\labelyv},
		ymax = 0.3,
		xmax = 300,
		xmin =1,
		ymin=0
	]
	% use TeX as calculator:
	\addplot+[red,solid,mark=none] table[x=f,y=vrf1] {simulationVel.txt}; \addlegendentry{\small RF}
	\addplot+[green,solid,mark=none] table[x=f,y=vci1] {simulationVel.txt}; \addlegendentry{\small CI}
	\addplot+[blue,solid,mark=none] table[x=f,y=vcf1] {simulationVel.txt}; \addlegendentry{\small CKF}
	\end{axis}
\end{tikzpicture}
}
\\
\rotatebox{90}{\hspace{1.0cm} Agent 2}
&
%%%%%%%%%%%%%%%%%%%%%%%%%%%%%%%%%%%%%%%%%%%%%%%%%%%%%%%%%%%%%%%%%%%%%%%%%%%%%%%
\scalebox{\figscale}{
\begin{tikzpicture}
	\begin{axis}[
		xlabel={\labelx},
		ylabel={\labely},
		ymax = 7,
		xmax = 300,
		xmin =1,
		ymin=0
	]
	% use TeX as calculator:
	\addplot+[red,solid,mark=none] table[x=f,y=prf2] {simulationPos.txt}; \addlegendentry{\small RF}
	\addplot+[black,solid,mark=none] table[x=f,y=pnf2] {simulationPos.txt}; \addlegendentry{\small NF}
	\end{axis}
\end{tikzpicture}
}
&
\scalebox{\figscale}{
\begin{tikzpicture}
	\begin{axis}[
		xlabel={\labelx},
		ylabel={\labely},
		ymax = 2,
		xmax = 300,
		xmin =1,
		ymin=0
	]
	% use TeX as calculator:
	\addplot+[red,solid,mark=none] table[x=f,y=prf2] {simulationPos.txt}; \addlegendentry{\small RF}
	\addplot+[blue,solid,mark=none] table[x=f,y=pcf2] {simulationPos.txt}; \addlegendentry{\small CKF}
	\end{axis}
\end{tikzpicture}
}
&
\scalebox{\figscale}{
\begin{tikzpicture}
	\begin{axis}[
		xlabel={\labelx},
		ylabel={\labely},
		ymax = 2,
		xmax = 300,
		xmin =1,
		ymin=0
	]
	% use TeX as calculator:
	\addplot+[red,solid,mark=none] table[x=f,y=prf2] {simulationPos.txt}; \addlegendentry{\small RF}
	\addplot+[green,solid,mark=none] table[x=f,y=pci2] {simulationPos.txt}; \addlegendentry{\small CI}
	\end{axis}
\end{tikzpicture}
}
&
\scalebox{\figscale}{
\begin{tikzpicture}
	\begin{axis}[
		xlabel={\labelx},
		ylabel={\labelyv},
		ymax = 0.3,
		xmax = 300,
		xmin =1,
		ymin=0
	]
	% use TeX as calculator:
	\addplot+[red,solid,mark=none] table[x=f,y=vrf2] {simulationVel.txt}; \addlegendentry{\small RF}
	\addplot+[green,solid,mark=none] table[x=f,y=vci2] {simulationVel.txt}; \addlegendentry{\small CI}
	\addplot+[blue,solid,mark=none] table[x=f,y=vcf2] {simulationVel.txt}; \addlegendentry{\small CKF}
	\end{axis}
\end{tikzpicture}
}
%%%%%%%%%%%%%%%%%%%%%%%%%%%%%%%%%%%%%%%%%%%%%%%%%%%%%%%%%%%%%%%%%%%%%%%%%%%%%%%
\\
\rotatebox{90}{\hspace{1.0cm} Agent 3}
&
%%%%%%%%%%%%%%%%%%%%%%%%%%%%%%%%%%%%%%%%%%%%%%%%%%%%%%%%%%%%%%%%%%%%%%%%%%%%%%%
%%%%%%%%%%%%%%%%%%%%%%%%%%%%%%%%%%%%%%%%%%%%%%%%%%%%%%%%%%%%%%%%%%%%%%%%%%%%%%%
\scalebox{\figscale}{
\begin{tikzpicture}
	\begin{axis}[
		xlabel={\labelx},
		ylabel={\labely},
		ymax = 7,
		xmax = 300,
		xmin =1,
		ymin=0
	]
	% use TeX as calculator:
	\addplot+[red,solid,mark=none] table[x=f,y=prf3] {simulationPos.txt}; \addlegendentry{\small RF}
	\addplot+[black,solid,mark=none] table[x=f,y=pnf3] {simulationPos.txt}; \addlegendentry{\small NF}
	\end{axis}
\end{tikzpicture}
}
&
\scalebox{\figscale}{
\begin{tikzpicture}
	\begin{axis}[
		xlabel={\labelx},
		ylabel={\labely},
		ymax = 2,
		xmax = 300,
		xmin =1,
		ymin=0
	]
	% use TeX as calculator:
	\addplot+[red,solid,mark=none] table[x=f,y=prf3] {simulationPos.txt}; \addlegendentry{\small RF}
	\addplot+[blue,solid,mark=none] table[x=f,y=pcf3] {simulationPos.txt}; \addlegendentry{\small CKF}
	\end{axis}
\end{tikzpicture}
}
&
\scalebox{\figscale}{
\begin{tikzpicture}
	\begin{axis}[
		xlabel={\labelx},
		ylabel={\labely},
		ymax = 2,
		xmax = 300,
		xmin =1,
		ymin=0
	]
	% use TeX as calculator:
	\addplot+[red,solid,mark=none] table[x=f,y=prf3] {simulationPos.txt}; \addlegendentry{\small RF}
	\addplot+[green,solid,mark=none] table[x=f,y=pci3] {simulationPos.txt}; \addlegendentry{\small CI}
	\end{axis}
\end{tikzpicture}
}
&
\scalebox{\figscale}{
\begin{tikzpicture}
	\begin{axis}[
		xlabel={\labelx},
		ylabel={\labelyv},
		ymax = 0.3,
		xmax = 300,
		xmin =1,
		ymin=0
	]
	% use TeX as calculator:
	\addplot+[red,solid,mark=none] table[x=f,y=vrf3] {simulationVel.txt}; \addlegendentry{\small RF}
	\addplot+[green,solid,mark=none] table[x=f,y=vci3] {simulationVel.txt}; \addlegendentry{\small CI}
	\addplot+[blue,solid,mark=none] table[x=f,y=vcf3] {simulationVel.txt}; \addlegendentry{\small CKF}
	\end{axis}
\end{tikzpicture}
}
%%%%%%%%%%%%%%%%%%%%%%%%%%%%%%%%%%%%%%%%%%%%%%%%%%%%%%%%%%%%%%%%%%%%%%%%%%%%%%%
\\
\rotatebox{90}{\hspace{1.0cm} Agent 4}
&
%%%%%%%%%%%%%%%%%%%%%%%%%%%%%%%%%%%%%%%%%%%%%%%%%%%%%%%%%%%%%%%%%%%%%%%%%%%%%%%
\scalebox{\figscale}{
\begin{tikzpicture}
	\begin{axis}[
		xlabel={\labelx},
		ylabel={\labely},
		ymax = 7,
		xmax = 300,
		xmin =1,
		ymin=0
	]
	% use TeX as calculator:
	\addplot+[red,solid,mark=none] table[x=f,y=prf4] {simulationPos.txt}; \addlegendentry{\small RF}
	\addplot+[black,solid,mark=none] table[x=f,y=pnf4] {simulationPos.txt}; \addlegendentry{\small NF}
	\end{axis}
\end{tikzpicture}
}
&
\scalebox{\figscale}{
\begin{tikzpicture}
	\begin{axis}[
		xlabel={\labelx},
		ylabel={\labely},
		ymax = 2,
		xmax = 300,
		xmin =1,
		ymin=0
	]
	% use TeX as calculator:
	\addplot+[red,solid,mark=none] table[x=f,y=prf4] {simulationPos.txt}; \addlegendentry{\small RF}
	\addplot+[blue,solid,mark=none] table[x=f,y=pcf4] {simulationPos.txt}; \addlegendentry{\small CKF}
	\end{axis}
\end{tikzpicture}
}
&
\scalebox{\figscale}{
\begin{tikzpicture}
	\begin{axis}[
		xlabel={\labelx},
		ylabel={\labely},
		ymax = 2,
		xmax = 300,
		xmin =1,
		ymin=0
	]
	% use TeX as calculator:
	\addplot+[red,solid,mark=none] table[x=f,y=prf4] {simulationPos.txt}; \addlegendentry{\small RF}
	\addplot+[green,solid,mark=none] table[x=f,y=pci4] {simulationPos.txt}; \addlegendentry{\small CI}
	\end{axis}
\end{tikzpicture}
}
&
\scalebox{\figscale}{
\begin{tikzpicture}
	\begin{axis}[
		xlabel={\labelx},
		ylabel={\labelyv},
		ymax = 0.3,
		xmax = 300,
		xmin =1,
		ymin=0
	]
	% use TeX as calculator:
	\addplot+[red,solid,mark=none] table[x=f,y=vrf4] {simulationVel.txt}; \addlegendentry{\small RF}
	\addplot+[green,solid,mark=none] table[x=f,y=vci4] {simulationVel.txt}; \addlegendentry{\small CI}
	\addplot+[blue,solid,mark=none] table[x=f,y=vcf4] {simulationVel.txt}; \addlegendentry{\small CKF}
	\end{axis}
\end{tikzpicture}
}
%%%%%%%%%%%%%%%%%%%%%%%%%%%%%%%%%%%%%%%%%%%%%%%%%%%%%%%%%%%%%%%%%%%%%%%%%%%%%%%

\end{tabular}
\caption{ {Comparison of the four methods}.   The Naive Fusion estimator quickly diverges, whereas the Robust Fusion and Covariance Intersection do not. 
Clearly the Robust Fusion estimator is more accurate than the Covariance Intersection and produces less noisy estimates due to its less conservative nature.
The steady state error covariance estimate of the Covariance Intersection is much larger than the actual error covariance making the estimator more susceptible to measurement noise.
}
\end{center}
\end{figure*}
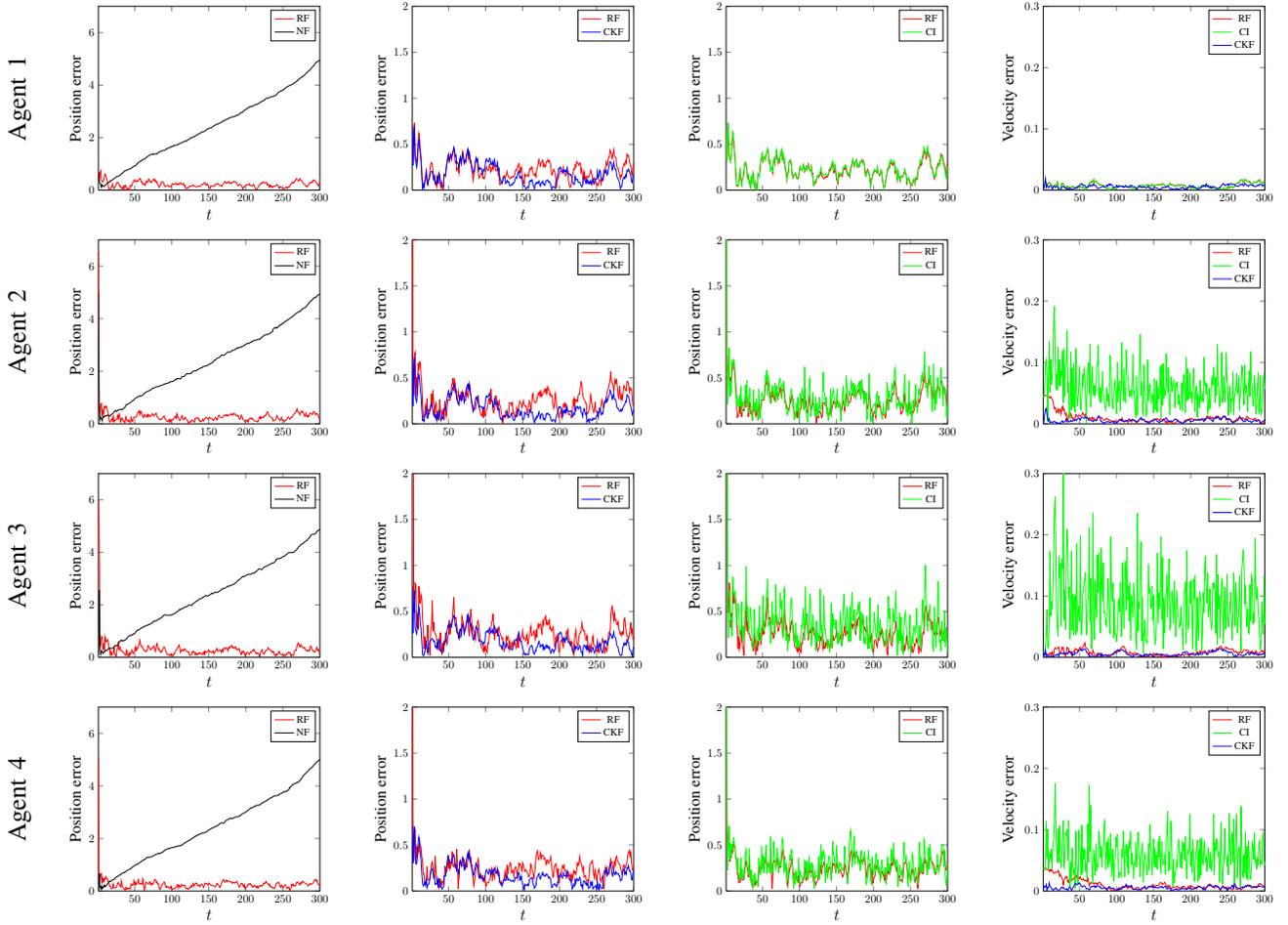

% steady state errors frame 51-300
% mean +- std 
% rf errors
% mean error for each agent
% agent 1:  0.2224 +-  0.0907
% agent 2:  0.2479 +-   0.1084
% agent 3:  0.2475 +- 0.1139
% agent 4:  0.2375 +- 0.0898

% cf errors
% mean error for each agent
% agent 1:  0.1737 +- 0.1069
% agent 2:  0.1664 +-  0.0977
% agent 3:  0.1697 +- 0.0988
% agent 4:  0.1610 +- 0.0851

% ci errors
% mean error for each agent
% agent 1:  0.2299 +- 0.0932
% agent 2:  0.2863 +- 0.1453
% agent 3:  0.3434 +- 0.1874
% agent 4:  0.2851 +-  0.1275

\begin{table}[ht!]
\caption{Position errors} % steady state  % title of Table
\label{tab:steady_state_errors}
\centering % used for centering table
\begin{tabular}{|c|c|c|c|} % centered columns (4 columns)
%\hline
\hline 
Agent $\#$ & CKF & RF  & CI\\ [0.2ex] 
\hline  \hline
 1 &   $0.174 \pm 0.107$ m &  $0.222 \pm  0.091$  m   & $0.230 \pm 0.093$ m \\
 2 &  $0.166 \pm  0.098$ m &  $0.248 \pm  0.108$  m  &  $0.286 \pm 0.145$ m \\
 3 &  $0.170 \pm 0.099$  m&  $0.248 \pm 0.114 $   m  & $0.343 \pm 0.187$ m \\
 4 &  $0.161 \pm 0.085$  m &  $0.238 \pm 0.090 $  m  & $0.285 \pm  0.128$ m \\
[0.2ex] % [1ex] adds vertical space
\hline %inserts single line
\end{tabular}
\end{table}

%%%%%%%%%%%%%%%%%%%%%%%%%%%%%%%%%%%%%%%%%%%%%%%%%%%%%%%%%%%%%%%%%%%%%%%%%%%%%%%%
\begin{comment}
\section{CONCLUSIONS AND FUTURE WORK}

\textcolor{red}{
A conclusion section is not required. Although a conclusion may review the main points of the paper, do not replicate the abstract as the conclusion. 
A conclusion might elaborate on the importance of the work or suggest applications and extensions. 
}
In the future, we plan to explore potential applications in SLAM and particularly, in multi-robot SLAM.
\end{comment}

\section{APPENDIX}

\subsection{Proof of lemma \ref{lem:simplefusion}}
\label{app:proofsimplefusion}

 First, it is easy to see that if $E[\widetilde{x}] = E[\widetilde{y}] =\mean{z}$  then $E[\widetilde{z}]=\mean{z}$.
 Now, one has to show that if $\Sxx \succeq \oSxx$ and $\Syy \succeq \oSyy$ then $ \trace(\Szz^\star) \geq \trace(\oSzz)$.
 We have that $\Szz^\star-\oSzz $ is equal to
 \begin{align*}
  & \ \  \ \ (I-K^\star)(\Sxx-\oSxx)(I-K^{\star T}) + K^\star(\Syy-\oSyy)K^{\star T} \\
  &+ K^\star (\Sxy^{\star T} - \oSxy^T)(I-K^{\star T}) + (I-K^\star) (\Sxy^{\star} - \oSxy)K^{\star T} \\
  & \succeq  K^\star (\Sxy^{\star T} - \oSxy^T)(I-K^{\star T}) + (I-K^\star) (\Sxy^{\star} - \oSxy)K^{\star T}
 \end{align*}
 since $\Sxx\succeq \oSxx$ and $\Syy\succeq \oSyy$.
Since trace is $\posdefcone{n}$-nondecreasing, we get
\[
\trace(\Szz^\star-\oSzz) \geq 2 \trace \left(K^T (I-K) (\Sxy^{\star} - \oSxy) \right) \geq 0
\]
since $\trace \left(K^{\star T} (I-K^\star) \Sxy^{\star} \right) \geq  \trace \left(K^{\star T} (I-K^\star) \Sxy \right)$ for every $\Sxy$  satisfying \eqref{eq:validcor} due to optimality of $\Sxy^\star$. 
Verifying that $\oSxy$ satisfies \eqref{eq:validcor} is straightforward.
\QEDA

\begin{comment}
\subsection{Proof of lemma \ref{lem:rffusion}}
\label{app:proofrffusion}

 First, it is easy to see that if $E[\widetilde{x}] = \mean{x}$, $E[\widetilde{y}] = \mean{y}$  and $E[\eta] = 0$ then $E[\widetilde{x}^+]=E[\widetilde{x}] = \mean{x}$.
 Now, one has to show that if $\Sxx \succeq \oSxx$ and $\Syy \succeq \oSyy$ then $\trace(\Sxx^{+\star}) \geq \trace(\oSxx)$. We have that $\Sxx^{+ \star}-\oSxx $ is equal to
 \begin{align*}
  & \ \ (I-K^\star C)(\Sxx-\oSxx)(I-C^T K^{\star T} ) \\
  & + K^\star  D(\Syy-\oSyy)K^{\star T} D^T \\
  & - K^\star D  (\Sxy^{\star T} - \oSxy^T)(I-K^{\star T} C^T )  \\
  & - (I-K^\star C) (\Sxy^{\star} - \oSxy)K^{\star T} D^T  + K^\star \Sigma_\eta  K^{\star T}\\
  & \succeq  
   - K^\star D  (\Sxy^{\star T} - \oSxy^T)(I-K^{\star T} C^T )  \\
  & - (I-K^\star C) (\Sxy^{\star} - \oSxy)K^{\star T} D^T 
 \end{align*}
 since $\Sxx\succeq \oSxx$,  $\Syy\succeq \oSyy$ and $\Sigma_\eta \succ 0$.
Since trace is $\posdefcone{n}$-nondecreasing, we get
\[
\trace(\Sxx^{+\star} - \oSxx)\geq -2 \trace \left( (I-K^\star C) (\Sxy^{\star} - \oSxy)K^{\star T} D^T  \right) \geq 0
\]
since $\Sxy^{\star}$ is optimal over all $\Sxy$ stisfying \eqref{eq:validcor}. 
Verifying that $\oSxy$ satisfies \eqref{eq:validcor} is straightforward.
\QEDA

\end{comment}

\subsection{Formulas for computing the Newton steps}
\label{app:newton}

First of all, the differential of $f_1(Q)$ at the direction of $\Delta Q$ is given by
 \begin{equation}
  Df_1(Q)[\Delta Q] =  \Syy^{-1/2} \left( \Delta Q^T  \Sxx^{-1} Q + Q^T  \Sxx^{-1} \Delta Q \right)\Syy^{-1/2} 
 \end{equation}

For small $\Delta X$, we have the first order approximation \cite{boyd2004convex}:
\begin{equation}
  \log \det (X+\Delta X) \approx \log\det(X) + \trace (X^{-1} \Delta X)
 \end{equation}
   and thus, using the chain rule, we obtain
\begin{equation}
\nabla_Q \log \det ( - f_1(Q)) =  2 \Sxx^{-1} Q   \Syy^{-1/2} f_1(Q)^{-1} \Syy^{-1/2} 
 \end{equation}
Moreover, we have
\begin{equation}
\begin{split}\nabla_X f(X ,Q ) = &2 \bigl( 
\begin{bmatrix}
C & D \\
\end{bmatrix}  
  \begin{bmatrix}
\Sxx & Q \\ Q^T & \Syy
\end{bmatrix} 
\begin{bmatrix}
C^T \\ D^T \\
\end{bmatrix}
+\Sigma_\eta 
  \bigr) X \\
   -&2 (C\Sxx+DQ^T)
  \end{split}
\end{equation}  
 and 
\begin{equation}  
\nabla_Q f(X ,Q ) = 2 (C^T XX^T D-X^TD) 
\end{equation}

 Let $g_1(Q) \doteq \Syy^{-1/2} f_1(Q)^{-1} \Syy^{-1/2}$. Using
 \begin{equation}
 (X+\Delta X)^{-1} \approx X^{-1} -  X^{-1} \Delta X X^{-1} 
 \end{equation}
 for small $\Delta X$ and the chain rule, we obtain the following system of linear equations for $(\Delta X_{nt},\Delta Q_{nt})$:
 
 \begin{equation}
 \begin{split}
  &2t \bigl( 
\begin{bmatrix}
C & D \\
\end{bmatrix}  
  \begin{bmatrix}
\Sxx & Q \\ Q^T & \Syy
\end{bmatrix} 
\begin{bmatrix}
C^T \\ D^T \\
\end{bmatrix}
+\Sigma_\eta 
  \bigr) \Delta X_{nt} \\
  + &2t (   C \Delta Q_{nt} D^TX - D \Delta Q_{nt}^T (I-C^T X)) = - \nabla_X f_t(X ,Q ) 
   \end{split}
   \end{equation}
  and
  \begin{equation}
  \begin{split}
  & 2t (C^T \Delta X_{nt}X^T D -(I- C^T X) \Delta X_{nt}^T D) \\
   - &2   \Sxx^{-1}   Q   g_1(Q)  \bigl( \ \Delta Q_{nt}^T  \Sxx^{-1} Q + Q^T  \Sxx^{-1} \Delta Q_{nt} \bigr)         g_1(Q)  \\
    +  &2 \Sxx^{-1} \Delta Q_{nt}   g_1(Q) = - \nabla_Q f_t(X ,Q )
 \end{split}
   \end{equation}

\bibliographystyle{plain}
\bibliography{egbib}

\end{document}